\documentclass[11pt,letterpaper]{article}

\usepackage{amsmath,amsthm,nicefrac}
\usepackage{amsfonts, amstext}
\usepackage{bm}

\usepackage{typearea}
\typearea{14}
\usepackage[font={small,it}]{caption}

\usepackage{setspace}
\usepackage[compact]{titlesec}

\usepackage[shortlabels]{enumitem}
\usepackage[breaklinks]{hyperref}
\hypersetup{colorlinks=true,%
            citebordercolor={.6 .6 .6},linkbordercolor={.6 .6 .6},%
citecolor=blue,urlcolor=black,linkcolor=blue,pagecolor=black}

\usepackage[nameinlink]{cleveref}
\Crefname{algocf}{Algorithm}{Algorithms}
\crefname{algocfline}{line}{lines}
\Crefname{invariant}{Invariant}{Invariants}
\Crefname{claim}{Claim}{Claims}
\Crefname{subclaim}{Subclaim}{Subclaims}

\usepackage{epsfig}
\usepackage{amsthm,amssymb}
\usepackage{amsthm,amssymb}

\usepackage{mathrsfs}
\usepackage{xspace}
\usepackage{soul}
\usepackage{latexsym}
\usepackage{bbm}

\usepackage{accents}
\newcommand\thickbar[1]{\accentset{\rule{.4em}{.8pt}}{#1}}

\usepackage{framed}

\usepackage[dvipsnames]{xcolor}
\definecolor{DarkGray}{rgb}{0.66, 0.66, 0.66}
\definecolor{DarkPowderBlue}{rgb}{0.0, 0.2, 0.6}
\definecolor{fluorescentyellow}{rgb}{0.8, 1.0, 0.0}
\definecolor{cerulean}{rgb}{0.0, 0.48, 0.65}
\definecolor{bleudefrance}{rgb}{0.19, 0.55, 0.91}

\usepackage[ruled,vlined,linesnumbered,algonl]{algorithm2e}
\SetEndCharOfAlgoLine{}
\SetKwComment{Comment}{\footnotesize$\triangleright$\ }{}

\SetCommentSty{mycommfont}

\usepackage{thmtools,thm-restate}

\allowdisplaybreaks

\newcommand{\Int}{{\mathbf {I}}}

\newcounter{note}[section]

\sethlcolor{fluorescentyellow}

\newcommand{\initOneLiners}{%
    \setlength{\itemsep}{0pt}
    \setlength{\parsep }{0pt}
    \setlength{\topsep }{0pt}
}

\newtheorem{theorem}{Theorem}[section]
\newtheorem{lemma}[theorem]{Lemma}
\newtheorem{claim}[theorem]{Claim}

\theoremstyle{definition}

\theoremstyle{remark}

\makeatletter 
\renewcommand{\theinvariant}{(I\@arabic\c@invariant)}
\makeatother

\newcommand{\eps}{\varepsilon}
\newcommand{\sse}{\subseteq}

\newcommand{\cost}{{\textsf{cost}}}

\newcommand{\poly}{\operatorname{poly}}

\newcommand{\Gen}{{\mathsf{Generate}}\xspace}

\newcommand{\cI}{{\mathcal{I}}}

\newcommand{\cR}{\mathcal{R}}

\newcommand{\nf}{\nicefrac}

\newcommand{\junk}[1]{}
\newcommand{\eat}[1]{}

\newif\ifhideproofs

\ifhideproofs
\usepackage{environ}
\NewEnviron{hide}{}

\fi

\newcommand{\wtd}{\textsc {Weighted $k$-Server}\xspace}
\newcommand{\kser}{\textsc {$k$-Server}\xspace}
\newcommand{\page}{\textsc {Paging}\xspace}

\newcommand{\vc}{\textsc {Vertex Cover}\xspace}

\begin{document}
\title{Efficient Algorithms and Hardness Results \\ for the \wtd \ Problem}

\author{
{Anupam Gupta\thanks{Computer Science Department, Carnegie Mellon University, Pittsburgh, PA. Email: {\tt anupamg@cs.cmu.edu}. Supported in part by NSF awards CCF-1955785, CCF-2006953, and CCF-2224718.}}
\and
{Amit Kumar\thanks{Department of Computer Science and Engineering, IIT Delhi, New Delhi, India. Email: {\tt amitk@cse.iitd.ac.in}.}}
\and
{Debmalya Panigrahi\thanks{Department of Computer Science, Duke University, Durham, NC. Email: {\tt debmalya@cs.duke.edu}. Supported in part by NSF awards CCF-1750140 (CAREER) and CCF-1955703.}}
}

\maketitle

\begin{abstract}
  In this paper, we study the weighted $k$-server problem on the
  uniform metric in both the offline and online settings. We start
  with the offline setting. In contrast to the (unweighted) $k$-server
  problem which has a polynomial-time solution using min-cost flows,
  there are strong computational lower bounds for the weighted
  $k$-server problem, even on the uniform metric. Specifically, we
  show that assuming the unique games conjecture, there are no
  polynomial-time algorithms with a sub-polynomial approximation
  factor, even if we use $c$-resource augmentation for $c <
  2$. Furthermore, if we consider the natural LP relaxation of the
  problem, then obtaining a bounded integrality gap requires us to use
  at least $\ell$ resource augmentation, where $\ell$ is the number of
  distinct server weights. We complement these results by obtaining a
  constant-approximation algorithm via LP rounding, with a resource
  augmentation of $(2+\eps)\ell$ for any constant $\eps > 0$.

  \medskip
  In the online setting, an $\exp(k)$ lower bound is known for the
  competitive ratio of any randomized algorithm for the weighted
  $k$-server problem on the uniform metric. In contrast, we show that
   $2\ell$-resource augmentation  can bring the competitive
  ratio down by an exponential factor to only $O(\ell^2 \log
  \ell)$. Our online algorithm uses the two-stage approach of first
  obtaining a fractional solution using the online primal-dual
  framework, and then rounding it online. 
\end{abstract}

\pagenumbering{gobble}

\clearpage

\pagenumbering{arabic}

\section{Introduction}
\label{sec:intro}

The \kser problem is a foundational problem in online algorithms and
has been extensively studied over the last 30
years~\cite{onlinebook}. In this problem, there are a set of $k$ servers
that must serve requests arriving online at the vertices of an
$n$-point metric space. The goal is to minimize the total movement
cost of the servers. The \kser problem was defined by Manasse et al.~\cite{ManasseMS90}, who also showed a lower bound of $k$ on the competitive ratio of any deterministic algorithm for this problem. Koutsoupias and Papadimitriou~\cite{KoutsoupiasP95} gave a $(2k-1)$-compeititive algorithm for \kser.   There has been much progress in the recent past on obtaining randomized algorithms with polylogarithmic (in $k$ and $n$) competitive ratio~\cite{BansalBMN15, BubeckCLLM18, Lee18, BuchbinderGMN19}. 
The \wtd version of this problem, introduced by
Fiat and Ricklin~\cite{FiatR94}, allows the servers to have non-uniform
positive weights; the cost of moving a server is now scaled by its
weight. In this paper, we consider the \wtd problem on a uniform
metric, namely when all $n$ points of the metric space are at unit
distance from each other, which means that the cost of moving a server
between any two distinct points is simply the weight of the
server. Note that the corresponding unweighted problem for the uniform
metric % , where the servers have equal weight,
is the extensively studied \page problem~\cite{onlinebook}. Indeed, one of
the original motivations for studying the \wtd problem came from a
version of paging with non-uniform replacement costs for different
cache slots~\cite{FiatR94}. In the rest of this paper, we % when we talk of
% the \wtd problem, we 
will implicitly assume that the underlying metric space is a uniform
metric.

The original paper of Fiat and Ricklin~\cite{FiatR94} introducing the
\wtd problem (on uniform metrics) gave a deterministic algorithm with
a competitive ratio of about $2^{2^{2k}}$. They also showed a lower
bound of $(k + 1)!/2$ for deterministic algorithms. Chiplunkar and
Viswanathan~\cite{ChiplunkarV} improved this lower bound to
$(k+1)!-1$, and gave a randomized algorithm that is
$1.6^{2^k}$-competitive against \emph{adaptive} online adversaries;
this also implies a deterministic competitive ratio of $2^{2^{k+1}}$
using the simulation technique of Ben-David et
al.~\cite{Ben-DavidBKTW94}. Bansal, Elias, and
Koumutsos~\cite{BansalEK} showed that this competitive ratio is
essentially tight for deterministic algorithms by showing a lower
bound of $2^{2^{k-4}}$. They also gave a deterministic {\em work
  function algorithm} with a competitive ratio of
$2^{2^{k+O(\log k)}}$. If the number of distinct server weights is
$\ell$ and there are $k_j$ servers of weight $W_j$, then the
competitive ratio of their algorithm is
$\exp(O(\ell k^3 \prod_{j=1}^\ell (k_j + 1)))$, which is an exponential
improvement over the general bound when $\ell$ is a
constant. Unlike
the \kser and \page problems, it is unknown if randomization
qualitatively improves the competitive ratio for \wtd, although the
best known lower bound for randomized algorithms against oblivious
adversaries is only singly exponential in $k$~\cite{AyyadevaraC21} as
against the doubly exponential lower bound for deterministic
algorithms.

The above competitive ratios depend only on $k$, and are
  independent of the size $n$ of metric space. Moreover, the hard
  instances are for metric spaces with the number of points $n$ that
  are exponentially larger than the number of servers $k$. This is not
  a coincidence, since better algorithms exist for smaller values of
  $n$. Indeed, the \wtd problem can be modeled as a metrical task
  system, where each state $\omega$ is a configuration (specifying the
  location of each of the $k$ servers), and the distance between any
  two states $\omega, \omega'$ is the cost to move between the
  configurations. Since there are $N = n^k$ states, one can obtain an
  $n^k$-competitive deterministic algorithm~\cite{BLS}, and an
  $O(\poly(k \log n))$-competitive randomized algorithm against
  \emph{oblivious} adversaries~\cite{BBBT,BBN,BCLL,CoesterL}; all these algorithms
  use $\poly(n^{k})$ time to explicitly maintain and manipulate the
  entire metric space, and hence are not efficient.

  In this paper we ask: \emph{is it possible to get efficient (randomized)
  \emph{online} algorithms that have competitive ratios of the form
  $\poly(k \log n)$, or even better? Is it possible to get such approximation ratios
  even in the \emph{offline} setting?} We show that obtaining improved competitive or approximation ratios in polynomial time is possible, provided we allow for {\em resource augmentation} in the number of servers. 

  Resource augmentation in online algorithms has been widely studied in paging and scheduling settings~(see e.g. \cite{KalyanasundaramP95, SleatorT85}). It is often a much needed assumption that allows for obtaining bounded or improved competitive ratios for such problems. Bansal et al.~\cite{BansalEJK19} studied the \kser problem on trees under resource augmentation.

\subsection{Our Results}
\label{sec:our-results}

Our first result establishes computational hardness of approximating
the \wtd problem in the offline setting. Unlike \page or \kser, which
are exactly solvable offline in polynomial time, we show that under
the Unique Games conjecture, the \wtd problem cannot be approximated to any subpolynomial factor even when we allow $c$-resource augmentation for any constant $c < 2$.

\begin{restatable}[Hardness]{theorem}{Hardness}
\label{thm:hard2}
For any constant $\eps > 0$, it is UG-hard to obtain an $N^{\nf{1}{2}-\eps}$-approximation algorithm for \wtd  with two weight classes, even when we are allowed $c$-resource augmentation for any constant $c < 2$. Here $N$ represents the size of the input (including the request sequence length). 
\end{restatable}

Next, we show that the natural time indexed LP relaxation for \wtd (see~\ref{LP:tag}) has a
large integrality gap, unless we allow for a resource augmentation of
almost $\ell$, the number of distinct server weights.

\begin{restatable}[Integrality Gap]{theorem}{IntGap}
  \label{thm:int}
  For any constant $\varepsilon >
  0$, the integrality gap of the relaxation~\ref{LP:tag} for \wtd
  is unbounded, even with $(\ell-\varepsilon)$-resource augmentation.
\end{restatable}

 It is worth noting  that an optimal fractional solution to~\ref{LP:tag} can be easily rounded to give an $\ell$-approximation ratio with $\ell$-resource augmentation. Indeed, we know that for each request, there exists a weight class which services this request to an extent of at least $\nf{1}{\ell}$. We can then scale this fractional solution by a factor $\ell$ and reduce this problem to $\ell$ instances of standard \page problem. The integrality gap result shows that any rounding algorithm with bounded competitive ratio must incur almost $\ell$-resource augmentation. 
We complement this integrality gap result with our main technical result, which gives an offline 
$O(1/\eps)$-approximation with $(2+\eps)\ell$-resource augmentation,
for any $\eps\in (0, 1)$.

\begin{restatable}[Offline Algorithm]{theorem}{Offline}
  \label{thm:main}
  Let $\cI$ be an instance of \wtd with $k_j$ servers of weight $W_j$ for all $j \in [\ell]$. 
  For any $\eps \in (0, 1)$, there is a polynomial time algorithm for $\cI$ that uses at most $2(1+\eps)\ell \cdot k_j$ servers of weights $W_j$ for each $j \in [\ell]$ and has server movement cost at most $O(1/\eps)$ times the optimal cost of $\cI$.
\end{restatable}

Finally, we obtain an online algorithm for \wtd with $2\ell$-resource
augmentation. The competitive ratio of the online algorithm is
$O(\ell^2\log \ell)$. (In constrast to the offline setting, it is no
longer clear how to achieve an $\ell$-competitive algorithm even if
we augment the number of servers by a factor of $\ell$.)

\begin{restatable}[Online Algorithm]{theorem}{Online}
  \label{thm:online}
  Let $\cI$ be an instance of \wtd with $k_j$ servers of weight $W_j$ for all $j \in [\ell]$. 
  There is a randomized (polynomial time) online algorithm for $\cI$ that uses at most $2\ell k_j$ servers of weights $W_j$ for each $j \in [\ell]$ and has expected server movement cost at most $O(\ell^2 \log \ell)$ times the optimal cost of $\cI$.
\end{restatable}

Since $\ell \le k$, the competitive ratio of the 
online algorithm is $O(k^2\log k)$. This implies that an $O(\ell^2)$-resource augmentation 
achieves at least an exponential improvement in the 
competitive ratio of the \wtd problem. Moreover, 
by rounding the weights to powers of $2$, we can assume that
$\ell \le O(\log W)$, where $W$ is the aspect ratio of the 
server weights. Hence, the competitive ratio of the 
online algorithm is $O(\log^2 W \log \log W)$.
Finally, note that for $\ell = O(1)$,
the above theorem gives a $O(1)$-competitive online algorithm with 
$O(1)$-resource augmentation. This can be seen as a generalization
of the classic result for the \page problem that achieves
a randomized competitive ratio of $O(\log \frac{k}{k-h+1})$
where the algorithm's cache has $k$ slots while the 
adversary's has only $h < k$ slots~\cite{Young91}.

\subsection{Our Techniques}
\label{sec:techniques}

In this section, we give an overview of the main techniques in the paper. The UG hardness of \wtd is based on a reduction from the \vc problem. Given an instance of the vertex cover problem, the corresponding \wtd consists of one point in the uniform metric space for each vertex of the graph.  The main observation is that if we know the minimum vertex cover size, we can keep one heavy weight server at each point corresponding to this vertex cover, which never change their positions. One can then generate an input sequence where the optimal solution pays a small cost, whereas an algorithm which does not cover an edge using heavy servers pays a much higher cost. The UG-hardness result for \vc translates to a corresponding resource augmentation lower bound for \wtd. Extending this approach to more than two weight classes (with stronger lower bounds on resource augmentation) turns out to be more challenging because the length of the input sequence becomes exponential in $n$. Instead, we show that the natural LP relaxation has a large integrality gap. The large gap instance consists of cycling through a sequence of subsets of the metric spaces with carefully varying frequency. The fractional solution is able to maintain a low cost by uniformly spreading servers over such cycles, but the integral solution is forced to service some of the cycles by small number of servers only. 

Our main technical result shows how to round a solution to the LP relaxation. The relaxation has both covering and packing type constraints, and the rounding carefully addresses one set of constraints without violating the other. We first scale the LP by a factor of about $2 \ell$, thus increasing both the resource augmentation and the cost. As a result, each request $\sigma_t$ is covered to an extent of $2\ell,$  and we can split this coverage across those weight classes which cover $\sigma_t$ to an extent of at least 1. Now for a fixed weight class, we consider the requests which are covered by it to an extent of at least 1. We show how to integrally round this solution so that this coverage property is satisfied and yet, we do not violate any packing constraint. After this, we show that the packing constraints can be ignored. This allows to scale down the LP solution by a factor $\ell$ (which saves the cost by this factor) and uses total unimodularity of the constraint matrix to round it. 

We extend our approximation algorithm to the online setting. The first step is to maintain an online fractional solution to the LP relaxation. Standard (weighted) paging algorithms for this problem rely on the fact that even the optimal offline algorithm needs to place a server at a requested location. But this turns out to be trickier here as we do not know the weight of the server which serves this location in the optimal solution. So we serve a request by ensuring that fractional mass from each weight classes is transferred at the same rate. The overall analysis proceeds by a careful accounting in the potential function. The online fractional solution satisfies the stronger guarantee that each request is served by servers of a particular weight class only. This allows us to reduce the rounding problem to independent instances of the \page problem.

We now give an overview of the rest of the submission.  In \S\ref{sec:int-gap}, we give details of the integrality gap construction; we defer the UG hardness proof to \S\ref{sec:hardness}. The offline rounding of the LP relaxation is given in \S\ref{sec:offline}, and then we extend this algorithm to the online case in \S\ref{sec:online}. 

\subsection{Preliminaries}
\label{sec:prelims}

In the \wtd problem on the uniform metric, we are given a set of $n$
points $V=\{1, \ldots, n\}$, such that $d(v,v') = 1$ for each
$v \neq v'$. There are $k$ servers, labeled $1, \ldots, k$, with
server $i$ having weight $w_i \geq 0$. The input specifies a request
sequence $(\sigma_1, \ldots, \sigma_T)$ of length $T$, with each
request $\sigma_t$ arriving at {\em time} $t$ being a point in $V$.  A
solution $f: [k] \times \{0,\ldots, T\} \to V$ specifies the position
of each server at each time $t \in [T]$ (where the initial positions
$f(i,0)$ are specified as part of the problem statement) such that for
each time $t$ there exists some server $i_t$ such that
$f(i_t,t) = \sigma_t$. The cost of the solution $f$ is the total
weighted distance travelled by the servers, i.e.,
$$ \nicefrac12 \sum_{i=1}^k w_i \; \sum_{t=1}^T \mathbbm{1}[f(i,t) \neq f(i,t-1)]. $$
The goal is to find a solution with the minimum cost. 
We say that an instance has $\ell$ {\em weight classes} if the set $\{w_1, \ldots, w_k\}$ has cardinality $\ell$. 
For an instance with $\ell$ different weight classes, we denote the
distinct weights by $W_1, \ldots, W_\ell$, and  let $k_j$ denote the
number of servers of weight $W_j$, with $\sum_j k_j = k$.
%Let $V_1, V_2, \ldots, V_\ell$ be the set of servers with weight $W_1,
%W_2, \ldots, W_\ell$ respectively.
For such an instance and a parameter $c \geq 1$, we say that the
algorithm uses \emph{$c$-resource augmentation} if it  uses $\lfloor
ck_j \rfloor$ servers of weight $W_j$ for each $j=1, \ldots, \ell$. 

We now describe the  natural LP relaxation for \wtd. It has a variable $x(v,j,t)$ for each request
time $t$, weight class $j \in [\ell]$ and vertex $v \in V$; it denotes
the fractional mass of servers of weight $W_j$ that are present at
point $v$ at time $t$. Let $\sigma_t$ denote the vertex requested at
time $t$. It is easy to verify that this is a valid relaxation. %Then the LP relaxation is as follows:
\begin{alignat}{2}
    \label{LP:tag}
    \tag{LP}
    \min \nicefrac12 \sum_{j \in [\ell]}  W_j  \sum_{t} \sum_{v \in V} &
    \; |x_{v,j,t}-x_{v,j,t-1}| && \\
    \sum_{v \in V} x_{v,j,t} & \leq k_j &\quad &\forall t, j \in [\ell]  \label{eq:cons1} \\
    \sum_{j \in [\ell]} x_{\sigma_{t},j,t} & \geq 1 & &\forall t \label{eq:cons2} \\
    x_{v,j,t} & \geq 0 && \forall t,v \in V,j \in [\ell] \notag
\end{alignat}

%%% Local Variables:
%%% mode: latex
%%% TeX-master: "main"
%%% End:

\section{The Unique Games Hardness}
\label{sec:hardness}

In this section, we consider the special case of \wtd when there are
only two weight classes. Assume wlog that the two distinct weights are
1 and $W$, where $W \gg 1.$ Our first main result shows that getting a
good approximation algorithm with $(2-\eps)$-resource augmentation for
any constant $\eps > 0$ is as hard as getting a better-than-two
approximation for the vertex cover problem. 

\Hardness*

\begin{proof}
  We give a reduction from the \vc problem. Let $d = d(\eps)$ be a
  constant to be fixed later, and let $c < 2$ be a constant as in the
  statement of the theorem.  Let $\cI = (G=(V,E), t)$ be an
  instance of the \vc problem on $n$ vertices. We know that it is
  UG-hard to distinguish between the following two cases: (i) $G$ has
  a vertex cover of size at most $t$, or (ii) every vertex cover of $G$
  must have size strictly larger than $ct$.
  % where $\delta > 0$ is a
  % positive constant.

  We map $\cI$ to an instance $\cI'$ of \wtd as follows: the set of
  points in $\cI'$ is given by $V \cup \{v_0\}$, where $v_0$ is a
  special vertex. There are $t$ servers of weight $W = n^{d}$ and one
  server of unit weight. Let the edges in $E$ be $e_1, \ldots, e_m$. A
  subsequence of the request sequence consists of $m$ {\em phases},
  where we have a phase for each edge $e_i$. During phase $i$
  corresponding to edge $e_i=(u_i, v_i)$, the request sequence toggles
  between $u_i$ and $v_i$ for $W$ times. Finally, the subsequence is
  repeated $W$ times.  In other words, it is the following sequence
  $$ \big( \underbrace{u_1, v_1, u_1, v_1, \ldots, u_1, v_1}_{W \
    \mbox{times}}, \ldots,\underbrace{u_m, v_m, u_m, v_m, \ldots, u_m,
    v_m}_{W \ \mbox{times}} \big)^W.  $$ We also have to specify the
  initial location of the servers. Assume that all servers are at
  location $v_0$ in the beginning.  This completes the description of
  the instance $\cI'$. Observe that $N$, the number of requests in
  instance $\cI'$ is $O(m \cdot n^{2d})$.

  \begin{claim}
    \label{cl:case1}
    Suppose $G$ has a vertex cover of size at most $t$. Then the cost of
    the optimal solution for $\cI'$ is at most $2mW$.
  \end{claim}
  \begin{proof}
    Let $V' \subseteq V$ be a vertex cover of size $t$. Consider the
    following solution: we move the $t$ heavy servers from $v_0$ to
    $V'$ at the beginning. From now on, the heavy servers will not
    move at all. During a phase corresponding to an edge
    $e_i= (u_i, v_i)$, we know that at least one of these vertices
    will be occupied by a heavy server. If the other end-point, say
    $v_i$, is not occupied by a heavy server, we move the server of
    weight 1 to $v_i$. Now we have two servers occupying $u_i$ and
    $v_i$ respectively until the end of this phase. The total movement
    cost is incurred either at the beginning (which is $tW$ overall),
    or at the beginning of each phase (when the cost is 1). Since
    there are $mW$ phases, the overall cost is at most
    $tW + mW \leq 2m W$.
  \end{proof}

  \begin{claim}
    \label{cl:case2}
    Suppose every vertex cover in $G$ has size strictly larger than
    $ct$. Then cost of the optimal solution for $\cI'$, even with
    $c$-resource augmentation, is at least $W^2$.
  \end{claim}
  \begin{proof}
    Consider any solution for $\cI'$. Recall that the input consists
    of $W$ subsequences, call these $S_1, \ldots, S_{W}$, where each
    subsequence $S_j$ consists of $m$ phases, one for each edge of
    $G$. We claim that during each such subsequence $S_j$, the
    solution must pay movement cost of at least $W$. Indeed, consider
    a subsequence $S_j$. If the solution moves a heavy server during
    $S_j$, then the claim follows directly. Else observe that the size
    of any vertex cover is strictly larger than the number of heavy
    servers $ct$, so there is some edge $e_i=(u_i, v_i)$ not covered
    by the heavy servers during $S_j$. Now the phase for $e_i$ in
    $S_j$ would require the unit weight server to toggle between $u_i$
    and $v_i$ for $W$ times. In either case, the cost of each
    subsequence is at least $W$, and the overall cost of the solution
    is at least $W^2$.
  \end{proof}

  The above two results along with the UG-hardness result for \vc
  impliy that it is UG-hard to obtain a
  $\frac{W^2}{2mW}$-approximation for \wtd with two weight
  classes. This ratio is equal to
  $\frac{W}{2m} \geq n^{d-2} \geq N^{\nf{1}{2}-\eps}$, assuming $d$ is
  $\Omega(\nf{1}{\eps})$, which proves \Cref{thm:hard2}.
\end{proof}

%%% Local Variables:
%%% mode: latex
%%% TeX-master: "main"
%%% End:

\section{An Integrality Gap for the Natural Linear Program}
\label{sec:int-gap}

In this section, we show that the relaxation~\ref{LP:tag} for \wtd has a
large integrality gap, unless we allow for a resource augmentation of
almost $\ell$, the number of distinct server weights.

Recall that the $\ell$ weights are denoted $W_1 \gg \cdots \gg
W_\ell$, and  there are $k_j$ servers of weight $W_j$. 
Our theorem is the following:

\IntGap*

\textbf{An Instance for Two Classes.}
To gain some intuition, we first consider the special case of
$\ell=2$, and prove the result without giving any resource
augmentation. There are $\nf{n}2$ servers of weight $W$ and $\nf{n}4$
servers of weight $1$, thereby giving a total of $k = \nf{3n}{4}$ servers. The input is given in ``phases''. Each phase is
specified by a distinct subset $S$ of $V$, where $|S|=\nf{n}2$. During
the phase corresponding to a subset $S$, we cycle through all subsets
$S'$ of $S$ with $|S'|= \nf{|S|}2 = \nf{n}4$. Given such a subset $S'$
of $S$, we send requests which cycle through the points in $S'$ for
$L$ times, where $L$ is large enough.

One fractional solution for such a sequence is defined as follows: we
assign $\nf12$ unit of weight-$W$ server at each of the $n$
locations. During the phase for a subset $S$, we assign $\nf12$ unit
of server of unit weight at each of the locations in $S$. The cost of
the fractional solution is at most
$Z := {n \choose n/2} \cdot \nf{n}4$ (not accounting for the initial
movement of the servers). However, an integral solution either moves
at least one heavy server, or else pays at least $L$ during one of the
phases, thereby must pay at least $\min(W, L)$. Assuming $W, L \gg Z$
gives an arbitrarily large integrality gap. (We can account for the
initial movement of the fractional servers by repeating the process
some $M$ times: the integral solution would pay at least $\min(W,L)$
in each such iteration and the fractional solution would pay at most
$Z$, so that the initial movement cost would get amortized away.)

\textbf{The Instance for $\ell$ Classes.}
We extend this construction to larger values of $\ell$ by defining
phases in a recursive manner on a nested sequence of subsets of $V$,
with each phase containing several repetitions of the same
sequence. Instead of decreasing by a factor 2, the number of servers
of each weight class now goes down by a factor of $C \geq \ell$. This
allows the integrality gap result to hold even when the integral
solution is allowed a resource augmentation of nearly $\ell$.

For some $r \leq \ell-1$, we call a tuple $(S_0, \ldots, S_r)$
\emph{valid} if (i)~$S_0 = V$ and each $S_j \sse S_{j-1}$, and
(ii)~$|S_j| = |S_{j-1}|/C = \nf{n}{C^j}$.
The request sequence is generated by calling \Cref{fig:proc} with the
trivial valid sequence $(S_0 = V)$. Given a valid tuple
$(S_0, \ldots, S_r)$, the procedure cycles through all subsets
$S \sse S_r$ of size $|S_r|/C$ and recursively calls
$\Gen(S_0, \ldots, S_r, S)$; this process is repeated $L_r$
times. Finally, in the base case when $r = \ell-1$, it cycles through
all the locations in $S_\ell$ for $L_{\ell-1}$ times.
For a suitably large choice of $M$, we define for each $r \in [\ell]$:
\begin{align} 
\label{eq:defWL}
  L_r := M^{r} \qquad \text{and} \qquad W_r := M^{\ell-r}.
\end{align}
Finally, the number of servers of weight $W_r$ is given by
$k_r := \frac{n}{\ell C^{r-1}}$.

\begin{algorithm}[H]
  \caption{Procedure $\Gen(S_0, S_1, \ldots, S_r)$. }
  \label{fig:proc}
  {\bf Input:} A valid tuple $(S_0, S_1, \ldots, S_r)$ \;
  \Repeat{$L_r$ iterations have occurred}{
    \label{l:repeat}
    \If{$r$ is equal to $\ell-1$}{
      Send a request at each location in $S_{\ell-1}$. \label{l:req} \;}
    \Else{
      \For{each subset $S$ of $S_r$ with $|S| = \frac{|S_r|}{C}$}{
        {\color{gray}// Move $\nf1\ell$ mass of servers of weight
        $W_{r+2}$ to $S$} \label{l:move}\;
        Call $\Gen(S_0, \ldots, S_r, S)$. \label{l:rec} \;
      }
    }
    \label{l:until}
  }
\end{algorithm}

\subsection{Analyzing the Integrality Gap}

We bound the cost of the optimal fractional solution for the above input sequence. 

\begin{lemma}
  \label{lem:fractionalcost}
  There is a fractional solution of total cost $O(f(n) M^{\ell-2})$
  for the input sequence constructed by~\Cref{fig:proc}, where $f(n)$
  is a function solely of $n$.
\end{lemma}
\begin{proof}
  Our fractional solution maintains the invariant: when the procedure
  $\Gen(S_0, \ldots, S_r)$ is called, we have $\nf1\ell$ fractional
  mass of servers of weight $W_1, \ldots, W_{r+1}$ respectively at
  each location in $S_r$. For the base case $r=0$, we place $\nf1\ell$
  server mass at each location in $S_0 = V$; recall that
  $k_1 = \nf{n}\ell$. For the inductive step, suppose this invariant
  is satisfied for a certain value of $r$ where $0 \leq r < \ell-1$;
  we need to show that it is satisfied for $r+1$ as well. Indeed, the
  induction hypothesis implies that we have $\nf1\ell$ amount of
  server mass of weight $W_1, \ldots, W_{r+1}$ at each location in
  $S_r$, and hence at each location in $S_{r+1}$.  Moreover, as
  line~\ref{l:move} indicates, we move $\nf1\ell$ fractional mass of
  servers of weight $W_{r+2}$ to each location in $S_{r+1}$ to satisfy
  the invariant condition.  This costs $W_{r+2} \, k_{r+2}/\ell$;
  moreover, this is feasible because the total number of servers of
  weight $W_{r+2}$ needed is
  $\frac{|S_{r+1}|}{\ell} = \frac{n}{\ell C^{r+1}} = k_{r+2}$.
  Finally, when $r=\ell-1$, the invariant shows that $1$ unit of
  server mass is present at each of the locations in $S_\ell$, and
  hence the requests generated in line~\ref{l:req} can be served
  without any additional movement of servers.

  We now account for the movement cost. The total server movement cost
  during $\Gen(S_0, \ldots, S_r)$ (not counting the movement costs in
  the recursive calls) is at most
  $O(L_r \, k_{r+1}\, W_{r+2}) = O(k_{r+1} \, M^{\ell-2})$.  Since
  $k_{r+1} \leq n$ and the number of calls to $\Gen$ is a function
  only of $n$, the overall movement cost can be expressed as
  $O(f(n) \cdot M^{\ell-2})$. (Again, by repeating the entire process
  multiple times we can amortize away the initial movement cost; we
  avoid this step for the sake of clarity.)
\end{proof}

The next lemma shows that any integral solution must have much higher
cost.

\begin{restatable}{lemma}{Integrality}
  \label{lem:integralcost}
  Let $\varepsilon > 0$ be a small enough constant. Assume that the
  integral solution is allowed $(\ell-\varepsilon) k_r$ servers of
  weight $W_r$ for each $r \in [\ell]$.  Any integral solution for the
  input sequence generated by~\Cref{fig:proc} (with
  $C = \nicefrac{\ell}{\varepsilon}$) has movement cost at least
  $M^{\ell-1}.$
\end{restatable}
\begin{proof}%[Proof of \Cref{lem:integralcost}]
  We prove the following more general statement by reverse induction
  on $r$: any integral solution for the sequence generated by
  $\Gen(S_0, \ldots, S_r)$ for a valid tuple $(S_0, \ldots, S_r)$
  which does not use any server of weight class $W_1, \ldots, W_r$ (at
  any location in $S_r$) has cost at least $M^{\ell-1}$. It suffices
  to prove this statement, because the case when $r=0$ implies the
  lemma.

  Consider the base case when $r=\ell-1$. Consider the sequence
  generated by such a procedure $\Gen(S_0, \ldots, S_r)$ such that no
  server of weight $W_1, \ldots, W_{\ell-1}$ is used for serving the
  requests at $S_{\ell-1}$.  Thus all requests generated by this
  procedure must be served by servers of weight $W_{\ell}$. Now,
  $|S_{\ell-1}| = \frac{n}{C^{\ell-1}}$, whereas the number of weight
  $W_{\ell}$ servers available to the algorithm is
  $(\ell-\varepsilon) k_\ell < \frac{n}{C^{\ell-1}}.$ Therefore,
  during each iteration of the {\bf repeat}-{\bf until} loop in
  lines~\ref{l:repeat}--\ref{l:until} in~\Cref{fig:proc}, at least one
  server of weight $W_\ell$ must move. So the overall movement cost
  during this input sub-sequence is at least
  $W_{\ell} \cdot L_{\ell-1} = M^{\ell-1}$. This proves the base case.

  The inductive case is proved in an analogous manner. Suppose the
  statement is true for $r+1$, and now consider the sub-sequence
  generated by $Gen(S_0, \ldots, S_r)$ for some valid tuple
  $(S_0, \ldots, S_r)$. Assume that no server of weight
  $W_1, \ldots, W_r$ is present at any node in $S_r$ during this
  time. We claim that the algorithm must incur movement cost of at
  least $W_{r+1}$ during each iteration of the {\bf repeat}-{\bf
    until} loop. Indeed, fix such an iteration. Two cases arise: (a)
  The algorithm moves a server of weight $W_{r+1}$ then the claim
  follows trivially, or (b) No server of weight $W_{r+1}$ is moved
  during this period. Now observe that $|S_r| = \frac{n}{C^r}$, and
  the number of weight $W_{r+1}$ servers available to the algorithm is
  $(\ell-\varepsilon) k_{r+1} = |S_r| - \varepsilon k_{r+1} = |S_r|
  \left(1 - \frac{1}{C} \right). $ Thus, there is a subset $S_{r+1}$
  of $S$ of size $\frac{|S_r|}{C} = \frac{n}{C^{r+1}}$ where no server
  of weight $W_{r+1}$ appears during this input sub-sequence. Consider
  the recursive call $\Gen(S_0, \ldots, S_r, S_{r+1})$ in
  line~\ref{l:rec}. The induction hypothesis implies that the movement
  cost during this recursive call is at least
  $M^{\ell-1} \geq W_{r+1}$.

  Thus, we have shown that the movement cost during each iteration of
  the {\bf repeat}-{\bf until} loop during $\Gen(S_0, \ldots, S_r)$ is
  at least $W_{r+1}$. Since there are $L_r$ such iterations, the
  overall movement cost is at least $W_{r+1} \cdot L_r = M^{\ell-1}.$
  This completes the proof of the induction hypothesis, and implies
  the lemma.
\end{proof}

%We defer the proof to \Cref{sec:appendixintegrality}; 
Combining \Cref{lem:fractionalcost} and \Cref{lem:integralcost} proves
\Cref{thm:int}.

%%% Local Variables:
%%% mode: latex
%%% TeX-master: "main"
%%% End:

\newcommand{\xI}{x_i^j(I)}
\newcommand{\yI}{y_{v,j,I}}
\newcommand{\yIs}{{\widetilde y}_{v,j,I}}
\newcommand{\yId}{{\thickbar y}_{v,j,I}}
\newcommand{\yIp}[1]{y_{#1,j,I}}
\newcommand{\xt}[1][v]{x_{#1,j,t}}
\newcommand{\xIp}[1]{x_{#1,j,I}}
\newcommand{\xIs}{\widetilde{x}_i^j(I)}
\newcommand{\xtminus}{x_{v,j,t^-}}
\newcommand{\xts}[1][v]{\widetilde{x}_{#1,j,t}}
\newcommand{\xtsp}[1]{{\widetilde{x}_{v,j,{#1}}}}
\newcommand{\xtsv}[1]{{\widetilde{x}_{#1,j,t}}}
\newcommand{\xtdv}[1]{{\thickbar{x}_{#1,j,t}}}
\newcommand{\xtd}{\thickbar{x}_{v,j,t}}
\newcommand{\xtminuss}{\widetilde{x}_{v,j,t^-}}
\newcommand{\xtdminus}{\thickbar{x}_{v,j,t^-}}
\newcommand{\tyI}{\widetilde{y}_{v,j,I}}
\newcommand{\yt}{y_i^j(t)}
\newcommand{\ytminus}{y_i^j(t^-)}
\newcommand{\yts}{\widetilde{y}_i^j(t)}
\newcommand{\ytminuss}{\widetilde{y}_{v,j,t^{-}}}

\newcommand{\RS}{{\mathsf {ScaleRound}}}
\newcommand{\halfeps}{\nicefrac{\eps}{2}}
\newcommand{\lastevent}{{\mathsf {LastEvent}}}
\newcommand{\downevent}{{\mathsf {DOWN}}\xspace}
\newcommand{\upevent}{{\mathsf {UP}}\xspace}
\newcommand{\lasttime}{{\mathsf {LastTime}}}
\newcommand{\tr}{{\mathsf{trunc}}}

\section{An Offline Algorithm via LP Rounding}
\label{sec:offline}

We now show an algorithm for the offline setting, that rounds any
fractional solution to the 
LP relaxation~\eqref{LP:tag}, and achieves the following guarantee:
\Offline*

Instead of working with the relaxation~\eqref{LP:tag}, we work with an
equivalent relaxation which turns out to be easier to interpret. For
each vertex $v \in V$, index $j \in [\ell]$ and time interval $I$, we
have a variable $y_{v,j,I}$, which denotes the fractional mass of
server of weight $W_j$ residing at $v$ during the entire time interval
$I$. The variable $x_{v,j,t}$ used in~\eqref{LP:tag} can be expressed as follows: 
\begin{align}
    \label{eq:xyrelation}
    x_{v,j,t} = \sum_{I: t \in I} y_{v,j,I}.
\end{align}
Let $\Int$ denote the set of all intervals during the request timeline. 
The new linear program relaxation for \wtd is the following:
\begin{alignat}{2}
  \label{lp:original}
   \tag{LP2}
  \min \nf12 \sum_{j \in [\ell] } W_j \; \sum_{I \in \Int} \sum_{v \in V}  &\yI  && \\
  \text{s.t.  } \sum_{j \in [\ell]} \sum_{I: t \in I}  \yIp{\sigma_t} &\ge 1 & \qquad \qquad
  &\forall t \label{eq:covLP} \\
  \sum_{v \in V} \sum_{I: t \in I}  \yI &\le k_j &&\forall  t, j \in [\ell] \label{eq:packLP}\\
  \notag \yI & \geq 0 \qquad \qquad &&  \forall t, j \in [\ell], v \in V. 
\end{alignat}
Note that the covering constraint~(\ref{eq:covLP}) enforces having at
least one unit of (fractional) server mass at the location $\sigma_t$
requested for each time $t$. The packing constraint~(\ref{eq:packLP})
enforces that the total (fractional) server mass of weight $W_j$ used
at any time $t$ is at most the number of servers of this weight,
namely $k_j$. Given a solution $\yI$ to~\ref{lp:original}, the
variables $\xt$ defined using~\eqref{eq:xyrelation} define a feasible
solution to~\ref{LP:tag} of the same cost.

Fix any constant $\eps\in (0, 1)$. We now prove~\Cref{thm:main} by rounding an optimal fractional solution $\yI$ to~\ref{lp:original}.  
% to obtain an integral solution that uses at most $(2+\eps)\ell k_j$ servers of weight $W_j$ and has cost of $O(1/\eps)$ times that of $\xI$. 
The rounding algorithm has two stages. The first stage scales and
discretizes the LP variables to integers such that
\begin{enumerate}[nosep]
\item the packing constraints are satisfied up to a factor of
  $(2+\eps)\ell$,
\item the covering constraints are satisfied with a
  scaled covering requirement of $\ell$ instead of 1, i.e.,
  $\sum_j \sum_{I: t\in I} \yIp{\sigma_t} \ge \ell,$ for all times $t$,
  and 
\item the cost of the fractional solution increases by a factor of
  $O(\ell/\eps)$.
\end{enumerate}
In the second stage, we remove the packing constraints from the LP;
this results in the resulting interval covering LP being
integral. Next, we scale the solution from the first stage down by
$\ell$, getting a feasible fractional solution to the standard LP
relaxation for the interval covering problem.
Finally, we use the integrality of the interval covering LP relaxation
to obtain an integral solution for~\ref{lp:original}. We present these
two stages in the next two sections.

\subsection{Stage I: Scaling and Discretization}

\begin{algorithm}[t]
  \caption{Procedure $\RS(x,y,v,W_j)$. }
  \label{algo:rs}
  {\bf Input:} A fractional solution $(\yI, \xt)$ to~\ref{lp:original}, a  location $v$ and a weight $W_j$\; 
  Initialize variables $\yId$ to 0 for all intervals $I$. \;
  {\bf (Scale):}   Define $\yIs = \left(2+\halfeps\right)\ell \cdot \yI$ and therefore,  \quad \quad $\xts = \sum_{I:t\in I} \yIs = \left(2+\halfeps\right)\ell \cdot \xt$ for each $I \in \Int$. \;
  {\bf (Round):} \For{$h =1, 2, \ldots, \ell$ \label{l:for}}{
   Initialize $\lastevent = \downevent, \lasttime = 0$. \;
   \Repeat{ we have reached the end of the timeline $[0,T]$}{
     \If{ $\lastevent = \upevent$}{
           Let $t$ be the first $\downevent$ after $\lastevent$ \;
           Update $\lastevent = \downevent, \lasttime = t. $
     }
     \Else{ ($\lastevent = \downevent$) Let $t$ be the first $\upevent$ after $\lastevent$ \;
        Add $I=[\lasttime,t)$ to $\Int_{v,j}(h)$ and increase $\yId$ by 1. \;
        Update $\lastevent = \downevent, \lasttime = t. $
     }
     }
      }
\end{algorithm}

The first stage of the rounding algorithm operates independently on
each location $v \in V$ and for each server weight $W_j$; the formal
algorithm $\RS(x,y,v,W_j)$ is given in~\Cref{algo:rs}. We work with
both the $y_{v,j,I}$ variables and the equivalent $\xt$ variables
defined in~\eqref{eq:xyrelation}; this representational flexibility
makes it convenient to explain the algorithm. To begin, we scale the
LP variables $\yI$ by a factor $(2+\halfeps)\ell$ to obtain $\yIs$ (we also
define the auxiliary variables $\xts$ by scaling $\xt$ similarly).

\emph{Discretization.} Next we discretize the scaled variables $\yIs$ and
$\xts$ to nonnegative integers $\yId$ and $\xtd$ respectively. To
start, let us describe the discretization of $\xts$ to obtain
$\xtd$. Intuitively, we would like to define $\xtd$ as
$\lfloor\xts\rfloor$, i.e., the largest step function with unit step
sizes entirely contained in $\xts$, but this can amplify small
fluctuations around integer values, and hence may increases the cost.
To avoid this, we introduce {\em hysteresis} in our discretization, by
setting different thresholds for increasing and decreasing the value
of $\xts$. We view $\xts$ as a time-varying profile and define horizontal {\em slabs} in it corresponding to the restriction of the range of $\xts$ to $[h,h+1)$ for some integer $h$. For each such slab, we identify intervals $I$ of width at most 1 and at least $\nf12$ and set  the increase the corresponding $\yId$ value by 1.
In more detail, for each such level $h$, we identify a subset
$\Int_{v,j}(h)$ of intervals for which the corresponding $\yId$
variable is to be increased by 1. 
We identify an alternating sequence of {\em up} and {\em down} events
in the timeline $[0,T]$ as follows:
\begin{itemize}
\item $\upevent$ event: At time $t$, there is an $\upevent$ event at
  level $h$ if $\xtminuss < h$ and $\xts \ge h$, and the previous
  event at level $h$ was a $\downevent$ event.
\item $\downevent$ event: At time $t$, there is a $\downevent$ event
  at level $h$ if the previous event at level $h$ was an $\upevent$,
  and $\xtminuss > h - \halfeps$ and $\xts \le h - \halfeps$, or
  $t=T$, the end of the timeline. (The reader should think of
  $\halfeps$ as the ``hysteresis gap'' between the up and down events
  at any level.)
\end{itemize}    
To make the definition complete, we set $\xts$ to 0 at $t = 0^-$ and
at $t=T^+$, and start with a $\downevent$ at time 0.  Finally, we add intervals stretching from each $\upevent$
to the next $\downevent$ to the set $\Int_{v,j}(h)$ of intervals. By
construction, these intervals are mutually disjoint. Finally, whenever
an interval $I$ is added to such a set $\Int_{v,j}(h)$, we increment
the corresponding variable $\yId$. Thus we have:
\[
    \yId = |\{h: I \in \Int_{v,j}(h) \}|, \text{ and correspondingly, } \xtd = \sum_{I: t\in I} \yId. 
\]

The next lemma shows that $\xtd$ can be thought of as a discretized form of $\xts$:
\begin{lemma}\label{lem:range}
The following holds for variables $\xtd$:
\begin{equation}\label{eq:range}
    \xts - 1 < \xtd < \xts + \halfeps.
\end{equation}
\end{lemma}
\begin{proof}
  Suppose $\xts \in [r,r+1)$. Consider the {\bf for} loop in
  line~\ref{l:for} in~\Cref{algo:rs} for a value $h \leq r$. We claim
  that at time $t$, the value of the variable $\lastevent$ must be
  $\upevent$. Suppose not. Let $t'$ be the value of $\lasttime$ at
  time $t$ (i.e., $t'$ is the last time before and including $t$ when
  an $\upevent$ or a $\downevent$ occurred). Since a $\downevent$
  event happened at time $t'$, $\xtsp{t'} < h$. Since $\xts \geq h$,
  an $\upevent$ event must occur during $(t',t]$, a
  contradiction. Therefore must have added an interval containing time
  $t$ to $\Int_{v,j}(h)$. Thus, $\xtd$ gets increased during each such
  iteration, i.e., $\xtd \geq r > \xts-1$. This proves the first
  inequality in~\eqref{eq:range}.

  We now prove the second inequality.
  Let $h$ be an integer satisfying $h \geq \xts + \halfeps.$ Consider
  the iteration of the {\bf for loop} in~\Cref{algo:rs} for this
  particular value of $h$. We claim that the value of the variable
  $\lastevent$ at time $t$ must be $\downevent$. Suppose not, and let
  $t'$ denote the value of the variable $\lasttime$. Then an
  $\upevent$ happened at time $t'$ and so $\xtsp{t'} \geq h$.  Since
  $\xtsp{t} \leq h-\halfeps$, a $\downevent$ event must have happened
  during $(t',t]$, a contradiction. Hence, we do not add any interval
  containing time $t$ to the set $\Int_{v,j}(h)$. Therefore,
  $\xtd < \xts + \halfeps$, which proves the second inequality
  in~\eqref{eq:range}.
\end{proof}

The next lemma establishes the key properties of the variables $\yId$ and $\xtd$.
\begin{lemma}\label{lem:discrete}
  The following properties hold the for the variables $\yId$:
  \begin{enumerate}[nosep,label=(\roman*)]
  \item (Cost) The LP cost increases by at most $O(\ell/\eps)$ when the original variables $\yI$ are replaced by the new variables $\yId$:
    \[
      \sum_{v, j, I} W_j \cdot \yId \le O(\ell/\eps) \cdot \sum_{v, j, I} W_j \cdot \yI.
    \]
  \item (Covering) The variables $\yId$ satisfy the scaled covering constraints of~\eqref{lp:original}
    \[
      \sum_{j, I: t \in I}   \yId \ge \ell \quad \forall t.
    \]           
  \item (Packing) The variables $\yId$ approximately satisfy the packing constraints of~\eqref{lp:original}:
    \[
      \sum_{v, I: t \in I} \yId \le (2+\eps)\ell k_j \quad \forall j \in [\ell], t.
    \]
  \end{enumerate}
\end{lemma}
\begin{proof}
  We first prove the cost bound: the cost of the solution $\yId$ is
  the weight of all intervals added to the sets $\Int_{v,j}(h)$
  for all $v,j,h$. I.e., 
  \begin{align}
    \label{eq:costy}
    \sum_{v, j, I} W_j \cdot \yId = \sum_{v,j} W_j \cdot \sum_{h \in [\ell]} |\Int_{v,j}(h)|. 
  \end{align}
  Fix a vertex $v$ and indices $j,h$.  For a non-negative number $x$,
  and non-negative integer $h$, define the \emph{$h$-level truncation}
  of $x$ to be $\tr_h(x):= \min(1, (x-h)^+)$, where $(a)^+ := \max(a,0)$ for any real $a$. Observe that
  $x = \sum_{h \geq 0} \tr_h(x)$.  In fact, for any two non-negative
  integers $x$ and $y$:
  \begin{align}
      \label{eq:trunc}
      |x-y| = \sum_{h' \geq 0} |\tr_{h'}(x)-\tr_{h'}(y)|. 
  \end{align}
  Now let $I_1=[s_1, t_1), \ldots, I_u=[s_u, t_u)$ be the intervals
  added to $\Int_{v,j}(h)$ (in left to right order). Define $t_0 =
  0$. We know that for any $i \in [u]$, an $\upevent$ happens at $s_u$
  and a $\downevent$ happens at $t_u$. Therefore,
  $\tr_h(\xtsp{s_u})- \tr_h(\xtsp{t_{u-1}}) \geq \halfeps$. Hence,
  \begin{align*}
    \varepsilon W_j/2 \cdot  |\Int_{v,j}(h)|
    & \leq W_J \cdot \sum_{i=1}^u |\tr_h(\xtsp{s_u})- \tr_h(\xtsp{t_{u-1}})| \\
    & \leq W_j \cdot  \sum_{t'=1}^T  |\tr_h(\xtsp{t-1})-\tr_h(\xts)|,
  \end{align*}
  where the last inequality follows from triangle inequality. Summing
  over all $h$ and using~\eqref{eq:trunc}, we get
  $$ \varepsilon   W_j/2 \cdot \yId \leq W_j \cdot  \sum_{t'=1}^T  |\xtsp{t-1})-\xts|. $$
  Summing over all vertices $v$ and indices $j \in [\ell]$, we see
  that the cost of the solution $\yId$ is at most $2/\varepsilon$
  times that of $\yIs$. Finally, the fact that $\yIs$ are obtained by
  scaling $\yI$ by a factor $(2+\halfeps) \ell$, we get the desired
  bound on the cost of $\yId$ solution.
  
  Next, we prove the covering property. Since $\xt$ is a feasible
  solution to~\ref{lp:original}, we have for any time $t$:
  \[
    \sum_j x_{\sigma_t, j, t} \ge 1, \text{ and therefore, } \sum_j \xtsv{\sigma_t}  \ge (2+\halfeps)\ell.
  \]
  Using \Cref{lem:range}, we have $\xtsv{\sigma_t} <  \xtdv{\sigma_t}
  + 1$, so
  \[
    \sum_{j \in \ell} \left(\xtdv{\sigma_t}+ 1\right) > (2+\halfeps)\ell, \text{ and therefore, }
    \sum_j \xtdv{\sigma_t} > \ell.
  \]
  
  Finally, we prove the packing property. Since $\xt$ is a feasible
  solution to the LP, we have for any $j \in [\ell]$ and time $t$,
  \[
    \sum_v \xt  \le k_j, \text{ and therefore, } \sum_v \xts  \le (2+\halfeps)\ell k_j.
  \]
  Again \Cref{lem:range} gives $\xts >  \xtd - \halfeps$, which implies
  \begin{equation}\label{eq:pack}
    \sum_j \left(\xtd - \halfeps\right)^+ < (2+\halfeps)\ell k_j.
  \end{equation}
  Since $\xtd$ is a
  nonnegative integer,
  \[
    \xtd > 0 \implies \xtd \ge 1 \stackrel{\small{\mbox{\Cref{lem:range}}}}{\implies} \xts > \xtd - \halfeps \ge 1 - \halfeps.
  \]
  Since $\sum_v \xts \le k_j$, it follows that the number of locations
  $v$ for which $\xtd > 0$ is at most $\frac{k_j}{1-\halfeps} < 2k_j$,
  if $\eps < 1$.  Using this fact in \Cref{eq:pack}, we get
\begin{align*}
\sum_v \xtd 
&= \sum_{v: \xtd > 0} \xtd 
= \sum_{v: \xtd > 0} \left(\xtd - \halfeps\right) + \sum_{v:\xtd > 0} \halfeps \\
&\le \sum_v \left(\xtd - \halfeps\right)^+ + 2k_j\cdot \halfeps 
\le (2+\halfeps) \ell k_j + \eps k_j.
\end{align*}
Since $\ell \ge 2$ (otherwise, we have the unweighted problem), we get
\[
\sum_v \xtd 
\le (2+\eps)\ell k_j. \qedhere
\]
\end{proof}

\subsection{Stage II: Weighted Interval Cover}

In the second stage of the rounding algorithm, we first scale the
(integer-valued) variables $\yId$ down by a factor of $\ell$ to obtain
new variables $\yI^*$:
\begin{gather}
  \yI^* := \yId/\ell \text{ and therefore, } \xt^* = \sum_{I:t\in I}
  \yI^* = \xtd/\ell. \label{eq:ystar}
\end{gather}
Our goal is to round the fractional variables $\yI^*$ to $\{0,1\}$
values.  In fact, our rounding will ensure that if the rounded value
equals $1$ then
$\yI^* > 0$. 
Since $\yId$ is
integral, the packing property in~\Cref{lem:discrete} implies that for
any time $t$, vertex $v$, and index $j \in [\ell]$, there are at most
$(2+\eps)\ell k_j$ intervals $I \ni t$ for which $\yId > 0$. The
rounding property of our algorithm will ensure that the final integral
solution, which lies in the support of $\yI^*$, will also satisfy that there are at most $(2+\varepsilon) \ell k_j$ intervals containing any time $t$. Since we are allowed a resource augmentation of
$(2+\varepsilon) \ell$ factor in the number of servers of weight
$W_j$, we can serve the requests with the set of available
servers. Henceforth, we can ignore the packing
constraint~\eqref{eq:packLP} for our rounded solution. As a result,
the relaxation~\ref{lp:original} decouples into $n$ independent
relaxations, one for each location $v \in V$.

In this decoupled instance, we get the following LP relaxation for
each location $v$, where for each class $j \in [\ell]$, we define
$\Int_{v,j}:=  \{I \mid \yI^* > 0\}$ as the set of intervals $I$ with a
nonzero value of $\yI^*$ and $\cR(v)$ as the set of times $t$ when $v$ is requested: 
\begin{alignat}{2}
  \label{lp:covering} \tag{LP$_v$}
  \min \nicefrac12 \sum_{j \in [\ell]}  W_j \cdot \sum_{I\in \Int_{v,j}}  &\yI \\
  \text{ s.t. } \sum_j \sum_{I\in \Int_{v,j}: t\in I} \yI &\ge 1
  &\quad \quad& \forall t \in \cR_v \notag\\
  \yI &\geq 0. \notag
\end{alignat}
By the covering property of \Cref{lem:discrete}, the variables $\yI^*$
defined in~(\ref{eq:ystar}) are feasible solutions for
\eqref{lp:covering} for all locations $v$. Furthermore, by the lemma's
cost property (and the scaling down by $\ell$), the total cost
$\sum_v \sum_j W_j \cdot \sum_I \yI^*$ is at most $O(1/\eps)$ times
the optimal cost of \eqref{lp:original}.

Finally, the constraint matrix for \eqref{lp:covering} satisfies the
consecutive-ones property: if the constraints are ordered
chronologically, then a variable $\yI$ appears in the constraints
corresponding to times $t\in I$ where $\sigma_t = v$, which is a
contiguous subsequence of all times $t$ where $\sigma_v =
j$. Constraint matrices with this property are totally unimodular
(see, e.g., \cite{FulkersonG65}). Therefore, each of the solutions $\{\yI^*: j \in [\ell], I \in \Int_{v,j} \}$ for \ref{lp:covering} can be rounded to a feasible integral solution without any increase in cost, 
which proves~\Cref{thm:main}.

%%% Local Variables:
%%% mode: latex
%%% TeX-master: "main"
%%% End:

\newcommand{\bigdot}[1] {\overset{\,_{\mbox{\Huge .}}}{#1}}
\newcommand{\zt}[1][v]{z_{#1,j,t}}
\newcommand{\zts}[1][v]{\widetilde{z}_{#1,j,t}}
\newcommand{\zst}[1][v]{z_{#1,\jstar,t}}
\newcommand{\dotzt}[1][v]{\dot{z}_{#1,j,t}}
\newcommand{\dotzst}[1][v]{\dot{z}_{#1,\jstar,t}}
\newcommand{\optt}[1][v]{\text{opt}_{#1,j,t}}
\newcommand{\ztminus}{z_{v,j,t^-}}
\newcommand{\opI}{p(I)}
\newcommand{\OPT}{{\cal O}}
\newcommand{\oqI}{q(I)}
\newcommand{\oyI}{y(I)}
\newcommand{\obyI}{\mathbf{y}(I)}
\newcommand{\obyIp}[1]{\mathbf{y}(#1)}
\newcommand{\jstar}{{j^\star}}

\section{Online Algorithm}\label{sec:online}

In this section, we describe an efficient online algorithm for \wtd and prove the following result:
\Online*

We begin by re-writing the LP relaxation~(\ref{lp:original}) in terms
of the ``anti-page'' variables, as in~\cite{BBN-focs07-paging}. Recall
that (\ref{lp:original}) has variables $\yI$ representing the
(fractional) weight $W_j$ server mass present at location $v$ during
the interval $I$; instead we first rewrite it in terms of the ``page''
variables $\xt$, which denote the total amount of weight $W_j$ server
mass at location $v$ at time $t$, as given in \eqref{eq:xyrelation}.
The objective of this LP in terms of $\xt$ is:
\[
     \sum_{v, j, I} W_j \cdot \yI = \sum_{v, j, I} W_j \cdot (\xt - \xtminus)^+.
\]
We can constrain any algorithm to values 
$\xt \in [0, 1]$ for all $v, j, t$ (since having multiple servers at a
location is not beneficial). This allows us to work with
non-negative \emph{anti-page} variables $\zt := 1 - \xt$. The
objective, now rewritten in terms of these new variables $\zt$, becomes:
\begin{gather}
  \sum_{v, j, I} W_j \cdot (\xt - \xtminus)^+ = \sum_{v,j,I} W_j \cdot
  (\ztminus - \zt)^+. \label{eq:online-obj}
\end{gather}
We shall also maintain the following invariant for each server weight $W_j$ and time $t$:
\begin{gather}
  \sum_v \xt = k_j \qquad \iff \qquad \sum_v \zt =  n - k_j~ \quad \forall j, t. \label{eq:online-cons1}
\end{gather}
We write the covering constraint~\eqref{eq:covLP} (or equivalently~\eqref{eq:cons2}) in terms of $\zt$ as:
\begin{equation}\label{eq:cover}
 \sum_j \zt[\sigma_t] \leq \ell-1
\end{equation}
The algorithm follows the standard relax-and-round paradigm in the
online setting. The first step is to compute a feasible fractional
solution to an LP consisting of objective (\ref{eq:online-obj}) and
constraints~(\ref{eq:online-cons1}) and~(\ref{eq:cover}), in an online
setting. We show in \S\ref{sec:online-fract-algor} that 
we can find a fractional solution that uses
$O(\ell k_j)$ servers of weight $W_j$ for each class $j$, and has a
competitive ratio of $O(\ell^2)$. The second step is to give an online
rounding algorithm to convert this fractional solution to an integral
solution: our rounding algorithm given in \S\ref{sec:online-rounding} uses 
the standard online rounding algorithm for the paging problem and 
increases the cost of the solution by a constant factor.

\subsection{Online Fractional Algorithm}
\label{sec:online-fract-algor}

In this section, we give an online algorithm for maintaining a
fractional solution to the LP involving $\zt$ variables. We obtain 
the following result:
\begin{theorem}\label{thm:fractional}
    There is a deterministic (polynomial time) online fractional algorithm that maintains 
    the condition that for every request time $t$, there exists an index $j \in [\ell]$ such that
    there is unit server mass of weight $W_j$ at location $\sigma_t$
    at time $t$.
    The algorithm uses $2\ell k_j$ servers of weight $W_j$ for each $j \in [\ell]$, 
    and whose cost is at most $O(\ell^2 \log \ell)$ times 
    that of an optimal fractional solution.
\end{theorem}

Note that the condition in the theorem is stronger than \eqref{eq:cover},
the feasibility condition for \eqref{lp:original}, because we are using server from a single weight class to service this request.

Consider a
time $t$, and the request arriving at location $\sigma_t$. We first
set $\zt = \ztminus$ for all $v \in V, j \in \ell$. Now the algorithm
moves fractional server mass to $\sigma_t$ until a relaxed version 
of the covering constraint~\eqref{eq:cover} for time $t$ gets
satisfied. The relaxed constraint is given by
\begin{equation}\label{eq:cover-relax}
\exists j \in [\ell] \text{ such that } \zt[\sigma_t] \leq 1-\frac{1}{2\ell}.
\end{equation}

Indeed, if the constraint is violated, then 
for each vertex $v \neq \sigma_t$ and each $j \in
[\ell]$, if $v$ has non-zero server mass of weight $W_j$ (i.e., $\zt < 1$),  
then the algorithm moves server mass 
of weight $W_j$ from $v$ to $\sigma_t$ using the following differential
equation. (The derivative is with respect to a variable $s$ which starts from 0 and increases at uniform rate.)
\begin{equation}\label{eq:rate}
    \dotzt = \frac{1}{W_j |S_j|}\cdot \left(\zt+\delta\right) ~\quad \forall j \in [\ell], \forall v \in S_j. 
\end{equation}
Here, $S_j \subseteq V$ denotes the instantaneous set of locations 
(i.e., at the current value of the
variable $s$) that have $\zt < 1$, not including the location
$\sigma_t$, and
$\delta > 0$ is a parameter that we shall fix later.
Correspondingly, we reduce $\zt[\sigma_t]$ by the total amount 
of server mass of weight $W_j$ entering $\sigma_t$:
\begin{equation}\label{eq:rate-sigmat}
    \dotzt[\sigma_t] = - \frac{1}{W_j |S_j|}\cdot \sum_{v\in S_j} \left(\zt+\delta\right) ~\quad \forall j \in [\ell]. 
\end{equation}
Note that server mass is moved away other locations and 
into location $\sigma_t$ only if 
$\zt[\sigma_t] > 1-\frac{1}{2\ell}$ for all $j$.
Since $\zt[\sigma_t]\le 1$ for all $j$, it follows that 
$\zt \in [1-\frac{1}{2\ell}, 1]$ for all $j, t$.
Hence,
\begin{equation}\label{eq:z-lb}
    \zt \ge 1-\frac{1}{2\ell} \text{ for all } j, t  \quad \implies \quad |S_j| \geq 2\ell k_j - 1 \ge \frac{3\ell k_j}{2} \ge 3 \text{ for all } j, t,
\end{equation}
since $\ell \ge 2, k_j \ge 1$.

To analyze the algorithm, we use a potential function $\Phi$. 
The potential function depends on the offline (integral) optimal
solution---let us call it $\OPT$, and let $\optt$ be the indicator
variable for the location $v$ 
containing a server of weight $W_j$ at time $t$. 
The potential at time $t$ is defined as follows:
\[
    \Phi(t) := \sum_{v, j: \optt = 0} W_j \cdot \ln \left(\frac{1+\delta}{\zt+\delta}\right).
\]
Let $\cost(t)$ denote the algorithm's server movement cost at time $t$ and $\cost^{\OPT}(t)$ denote the corresponding quantity for the optimum solution $\OPT$. Our goal is to show that: 
\begin{align}
    \label{eq:pot}
    \frac{\cost(t)}{4\ell} + \Phi(t+1)-\Phi(t) \leq  \ln(1+\nf1\delta) \cdot \cost^{\OPT}(t).
\end{align}

The following properties of $\Phi(t)$ can verified easily:
\begin{itemize}
    \item {\bf Nonnegativity:} $\Phi$ is always nonnegative, since
      $\zt \le 1$.
    \item {\bf Lipschitzness property:} When the optimal solution moves a server of weight $W_j$ from one location to another, the increase in $\Phi$ is at most $W_j\cdot \ln (1+\nf1\delta)$.
\end{itemize}
The Lipschitzness property implies that~\eqref{eq:pot} holds when
$\OPT$ serves the request at $\sigma_t$.  It remains the analyze the
cost and change in potential when the algorithm changes its
solution. Consider the process when we transfer server mass to
$\sigma_t$. 

We first bound the online algorithm's cost. Since all the weight classes incur the same server movement cost while transferring to $\sigma_t$, the movement cost is $\ell$ times the movement cost incurred while transferring servers of a fixed class, say $j^\star$. 
The latter is at most 
\begin{align}
    W_\jstar \sum_{v\in S_\jstar } \dotzst 
    &\stackrel{\eqref{eq:rate}}{=}  \frac{1}{|S_\jstar|} \sum_{v\in S_\jstar} (\zst+\delta)
    \quad = \frac{|S_\jstar|+1 -k_\jstar + \delta |S_\jstar|}{|S_\jstar|} \leq 1+\delta.
\end{align}
Thus, the  upper bound on the $\frac{\cost(t)}{4\ell}$ term in the LHS of~\eqref{eq:pot} is at most 
$\frac{1+\delta}{4} \leq \nf{1}{3}$
provided $\delta \le \nf 13$.

Next, we lower bound the rate of decrease of potential $\Phi$. 
We begin by bounding the rate of decrease in potential due to because of server mass leaving all locations 
except $\sigma_t$:
\begin{align}
    \Delta^- 
    &= - \sum_{j \in [\ell], v \neq \sigma_t: \optt = 0} \frac{W_j}{\zt + \delta} \cdot \dotzt
    \quad \stackrel{\eqref{eq:rate}}{=} - \sum_{j, v \in S_j: \optt = 0} \frac{1}{\zt+\delta} \cdot \frac{\zt+\delta}{|S_j|} \notag \\
    &= -\sum_j  \frac{|\{v \in S_j: \optt = 0 \}|}{|S_j|}  
    \quad \stackrel{\eqref{eq:z-lb}}{\le}  - \sum_j \frac{|S_j| - k_j}{|S_j|} \leq -\ell \left(1 - \frac{2}{3\ell} \right) 
    = -\ell + \nf 23. \label{eq:delta-minus}
\end{align}
Next, we bound the rate of increase in potential due
to server classes $j\not=j^*$ because of server mass entering
$\sigma_t$:
\begin{align*}
    \Delta^+ 
    &= \sum_{j \not= j^*} \frac{W_j}{\zt[\sigma_t]+\delta} \cdot \dotzt[\sigma_t]
    \quad \stackrel{\eqref{eq:rate}}{=} \sum_{j \not= j^*, v\in S_j} \frac{W_j}{\zt[\sigma_t]+\delta} \cdot \frac{\zt+\delta}{|S_j|W_j}\\
    &= \sum_{j \not= j^*}   \frac{\sum_{v\in S_j} (\zt+\delta)}{|S_j|(\zt[\sigma_t]+\delta)} 
    \quad = \sum_{j \not= j^*}  \frac{(|S_j| - k_j  + (1-\zt[\sigma_t]))+\delta \cdot |S_j|}{|S_j|(\zt[\sigma_t]+\delta)} \; \\ 
    &\stackrel{\eqref{eq:z-lb}}{\le} \sum_{j \not= j^*}  \frac{(|S_j| - k_j  + \nf1{2\ell})+\delta \cdot |S_j|}{|S_j|(1-\nf{1}{2 \ell}+\delta)} 
    \quad \stackrel{\eqref{eq:z-lb}}{\le} \sum_{j \neq \jstar} \frac{1 - \nf{2}{3\ell} + \nf{1}{6\ell} + \delta}{1 - \nf{1}{2\ell} + \delta} %\leq  \sum_{j \neq \jstar} \left( 1 - \nf{1}{2 \ell} \right) \cdot \left( 1 - \nf{1}{3\ell} \right) \\
    \quad\leq \ell-1, 
\end{align*}
provided $\delta = \nf{1}{2\ell}$. Combining with
\eqref{eq:delta-minus}, we see that the overall change in potential is
$\Delta^- + \Delta^+ \leq -\nf13$. 
Consequently, we get that the change in potential pays for the
increase in the algorithm's cost (divided by $4\ell$)---which shows
\eqref{eq:pot}---when the fractional solution changes.

This implies that we have an algorithm for maintaining $\zt$ that satisfies \eqref{eq:cover-relax}. In terms of the competitive ratio, the algorithm loses $4\ell$ in \eqref{eq:pot} and $\ln(1+\nf1{\delta}) = O(\log \ell)$ in the Lipschitzness of the potential function. Note that \eqref{eq:cover-relax} implies that for all $t$, there exists $j$ such that $\xt[\sigma_t] \geq \frac{1}{2\ell}$. We scale the fractional variables to obtain $\xts := \min(2\ell\xt, 1)$; then, for all $t$, there exists $j$ such that $\xts[\sigma_t] = 1$. Note that this satisfies the condition in \Cref{thm:fractional}. Equivalently, the corresponding ``anti-page'' variables $\zts := 1-\xts$ satisfy the following condition for all $t$:
\begin{equation}\label{eq:cover-zts}
    \exists j \text{ such that } \zts[\sigma_t] = 0.
\end{equation}
The last scaling step creates a resource augmentation of $2\ell$, and increases the competitive ratio to $O(\ell^2\log \ell)$. 
This completes the proof of \Cref{thm:fractional}.

\subsection{Rounding the Fractional Solution Online}
\label{sec:online-rounding}

We round the fractional solution for each weight class $j$ separately. 
Let $T_j$ represent the request times $t$ when \eqref{eq:cover-zts} is satisfied by 
weight class $j$. Note that the solution $\zts$ for weight class $j$
represents a feasible fractional solution for an instance of the paging
problem with $2\ell k_j$ cache slots,
where there is a page request for each time $t\in T_j$ at location $\sigma_t$.

We now invoke the following known online rounding algorithm for the paging problem
separately in each weight class $j$ to complete the proof of \Cref{thm:online}.
\begin{lemma}\label{lem:page-rounding}{\cite{BlumBK99}}
    There is a randomized (polynomial time) online algorithm that converts any feasible fractional solution
    for an instance of the \page problem to an integral solution using the 
    same number of cache slots, and incurs constant times the cost of the fractional solution.
\end{lemma}

%%% Local Variables:
%%% mode: latex
%%% TeX-master: "main"
%%% End:

\section{Discussion}

In this work, we have given the first efficient offline and online algorithms with non-trivial guarantees for \wtd. Several interesting problems remains open:
\begin{enumerate}
\item For the case of two distinct weight classes, we show in
  \Cref{sec:hardness} that it is
  UG-Hard to obtain an $\Omega(N^c)$-approximation algorithm for some
  constant $c>0$, even with $(2-\varepsilon)$-resource
  augmentation. Can we extend such a hardness result to more weight classes? For
  example, can we show that for three distinct weight classes, it is
  UG-Hard to obtain a $C$-approximation algorithm for any {\em
    constant} $C$, even with $(3-\varepsilon)$-resource augmentation?
  The principal reason why our hardness proof for $\ell=2$ does not
  extend here is because one needs to recursively cycle through all
  subsets (of a certain size) of $V$ to create an integrality gap
  instance for the natural LP relaxation. If the size of these subsets
  is large, then the length of the input becomes very large. If the
  size of these subsets is small, then it is not clear how to extend
  this to a hardness proof.
\item In \Cref{sec:offline}, we give an offline constant approximation algorithm which
  requires slightly more than $2 \ell$-resource augmentation. Can we
  get a constant approximation algorithm (or even an optimal
  algorithm) with exactly $\ell$-resource augmentation? We conjecture
  that the integrality gap of~\ref{LP:tag} is constant (or even $1$)
  if the integral solution is allowed $\ell$-resource augmentation.
\item In the online case, we give a $O(\ell^2 \log \ell)$-competitive
  algorithm with $2\ell$-resource augmentation in \Cref{sec:online}. Can we get a
  constant-competitive algorithm with $O(\ell)$-resource augmentation,
  i.e., a result in the same vein as our offline algorithm?
\end{enumerate}

%%% Local Variables:
%%% mode: latex
%%% TeX-master: "main"
%%% End:

{\small
\bibliographystyle{alpha}
\bibliography{bib}
}

\end{document}